\documentclass[a4paper,10pt]{article}

\usepackage{fullpage}
\usepackage[latin1]{inputenc}
\usepackage[T1]{fontenc}
\usepackage{setspace}
\usepackage[normalem]{ulem}
\usepackage{verbatim}
\usepackage{graphicx}
\usepackage{amssymb}
\usepackage{amsmath}
\usepackage{booktabs}
\usepackage{algorithmic}
\usepackage{algorithm}
\usepackage{tikz}
\usetikzlibrary{positioning}
\usetikzlibrary{decorations.pathreplacing}
\usepackage{amsthm}
\usepackage{subfigure}
\usepackage{authblk}
\usepackage{mathdots}
\newtheorem{lem}{Lemma}
\newtheorem{cor}{Corollary}
\newtheorem{theor}{Theorem}
\DeclareMathOperator{\im}{im}
 \DeclareMathOperator{\rank}{rk}

\usepackage{etoolbox}

\newcommand{\n}[1]{%
  \ifstrequal{#1}{1}{_{1}}{}
	\ifstrequal{#1}{2}{_{2}}{}
	\ifstrequal{#1}{3}{_{3}}{}
  \ifstrequal{#1}{4}{_{4}}{}
	\ifstrequal{#1}{i}{_i}{}
}

\newcommand{\lgen}[0]{\langle}
\newcommand{\rgen}[0]{\rangle}
\newcommand{\wht}[1]{H_{#1}}
\DeclareMathOperator{\dft}{DFT}
\DeclareMathOperator{\diag}{Diag}
\DeclareMathOperator{\gl}{GL}

\newcommand{\mypar}[1]{{\bf #1.}}

\title{Characterizing and Enumerating Walsh-Hadamard Transform Algorithms}
\author{Fran\c cois Serre}
\author{Markus P\"uschel}
\affil{Department of Computer Science, ETH Zurich}
\date{}
\begin{document}
\maketitle
\begin{abstract}
We propose a way of characterizing the algorithms computing a Walsh-Hadamard transform that consist of a sequence of arrays of butterflies ($I_{2^{n-1}}\otimes \dft_2$) interleaved by linear permutations. Linear permutations are those that map linearly the binary representation of its element indices. We also propose a method to enumerate these algorithms. 
\end{abstract}

\section{Introduction}
The \emph{Walsh-Hadamard transform} (WHT) is an important function in signal processing \cite{Beauchamp:85} and coding theory \cite{MacWilliams:77}. It shares many properties with the \emph{Discrete Fourier Transform} (DFT), including a Cooley-Tukey \cite{CT} divide-and-conquer method to derive fast algorithms. Pease-like \cite{Pease:68} WHT (Fig~\ref{fig:pease}) and the iterative Cooley-Tukey WHT (Fig~\ref{fig:itct}) are two examples of these. The algorithms obtained with this method share the same structure: a sequence of arrays of \emph{butterflies}, i.e., a block computing a DFT on $2$ elements, interleaved with \emph{linear permutations}. Linear permutations are a group of permutations appearing in many signal processing algorithms, comprising the bit-reversal, the perfect shuffle, and stride permutations.

In this article, we consider the converse problem; we derive the conditions that a sequence of linear permutations has to satisfy for the corresponding algorithm to compute a WHT (Theorem~\ref{trm:characterisation}). Additionally, we provide a method to enumerate such algorithms (Corollary~\ref{trm:enumerating}).

\begin{figure}\centering
\subfigure[Pease WHT]{\includegraphics[scale=1]{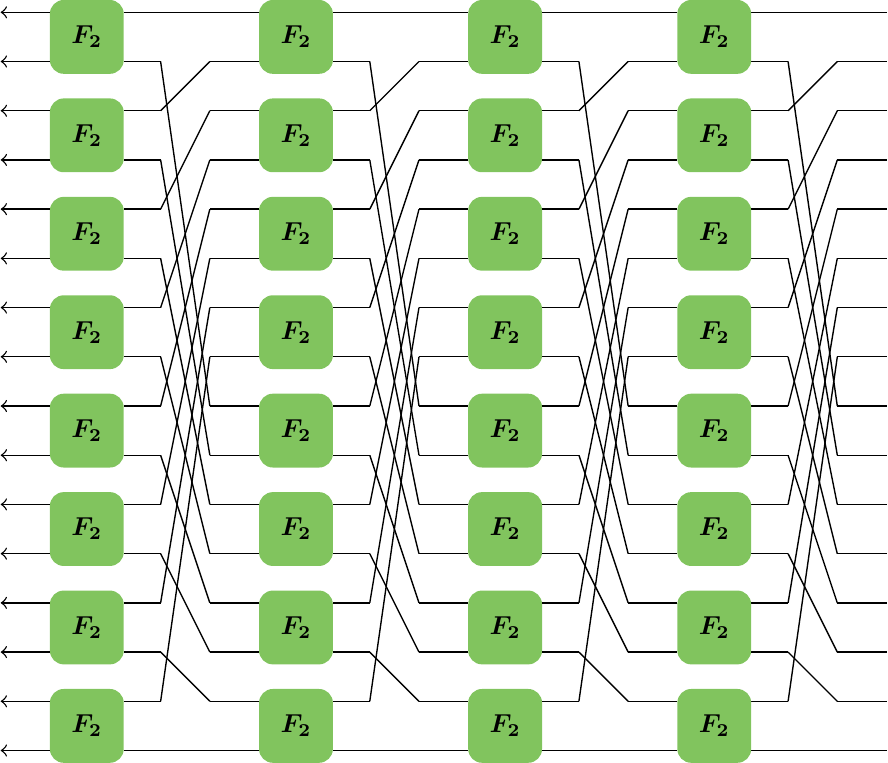}\label{fig:pease}}
\subfigure[Iterative Cooley-Tukey]{\includegraphics[scale=1]{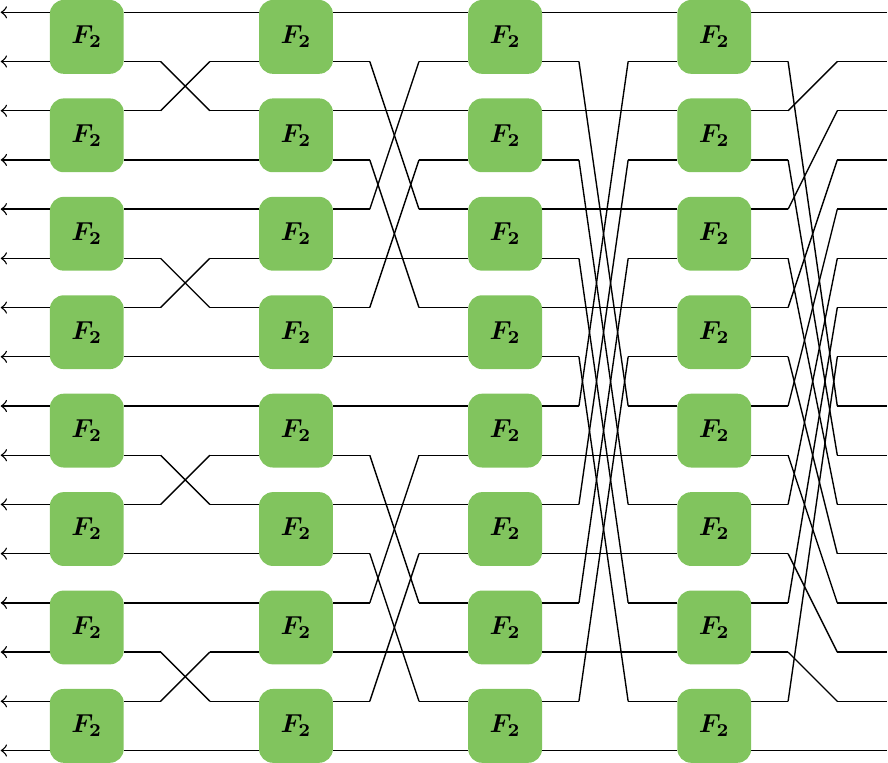}\label{fig:itct}}
\caption{Dataflow of two Cooley-Tukey derived fast algorithms computing a WHT on $16$ elements. The $F_2$ blocks represent butterflies.}
\end{figure}

\subsection{Background and notation}
\mypar{Hadamard matrix} For a positive integer $n$, and given $x\in\mathbb R^{2^n}$, the WHT computes $y=\wht{n}\cdot x$, where $\wht{n}$ is the \emph{Hadamard matrix} \cite{Hadamard:93}, the $2^n\times 2^n$ square matrix defined recursively by
\begin{equation*}
\wht{n}=\begin{cases}
\dft_2,&\text{ for n=1, and}\\
\left(\begin{array}{rr}\wht{n-1}&\wht{n-1}\\\wht{n-1}&-\wht{n-1}\end{array}\right),&\text{ for n>1.}\\
\end{cases}
\end{equation*}
Here, $\dft_2$ is the butterfly, i.e. the matrix
\[\dft_2=\begin{pmatrix}\begin{array}{rr}1&1\\1&-1\end{array}\end{pmatrix}.\]

For instance, the WHT on $8$ elements corresponds to the matrix
\[\wht{3}=\begin{pmatrix}\begin{array}{rrrrrrrr}
    1 &  1 &  1 &  1 &  1 &  1 &  1 &  1\\
    1 & -1 &  1 & -1 &  1 & -1 &  1 & -1\\
    1 &  1 & -1 & -1 &  1 &  1 & -1 & -1\\
    1 & -1 & -1 &  1 &  1 & -1 & -1 &  1\\ 
    1 &  1 &  1 &  1 & -1 & -1 & -1 & -1\\
    1 & -1 &  1 & -1 & -1 &  1 & -1 &  1\\
    1 &  1 & -1 & -1 & -1 & -1 &  1 &  1\\
    1 & -1 & -1 &  1 & -1 &  1 &  1 & -1
\end{array}\end{pmatrix}.
\]

\mypar{Binary representation} For an integer $0\leq i <2^n$, we denote with $i_b$ the column vector of $n$ bits containing the binary representation of $i$, with the most significant bit on the top ($i_b\in \mathbb F_2^n$, where $\mathbb F_2$ is the Galois field with two elements). For instance, for $n=3$, we have \[6_b=\begin{pmatrix}1\\1\\0\end{pmatrix}.\]
Using this notation, a direct computation shows that the Hadamard matrix can be rewritten as
\begin{equation}
\label{equ:def}
\wht{n}=\left((-1)^{i_b^Tj_b}\right)_{0\leq i,j < 2^n}.
\end{equation}

\mypar{Cooley-Tukey Fast WHT} Cooley-Tukey algorithms are based on the following identity, satisfied by the Hadamard matrix:
\begin{equation}
\label{equ:CT}
\wht{n}=\wht{p}\otimes \wht{q}=(\wht{p}\otimes I_{2^q})\cdot (I_{2^p}\otimes \wht{q})\text{, where } p+q=n.
\end{equation}

Using this formula recursively  -- along with properties \cite{Johnson:90} of the \emph{Kronecker product} $\otimes$ -- yields expressions consisting of $n$ arrays of $2^{n-1}$ butterflies,
\[
I_{2^{n-1}}\otimes \dft_2=\begin{array}{c}\begin{pmatrix}
\dft_2& &\\
&\ddots&\\
&&\dft_2
\end{pmatrix}\\
\underbrace{\rule{78pt}{0pt}}_{2^{n-1}\text{ times}}\end{array},
\]
interleaved by permutations. For instance, the Pease-like \cite{Pease:68} algorithm for the WHT (Fig.~\ref{fig:pease}) uses $n$ \emph{perfect shuffles}, a permutation that interleaves the first and second half of its input, and that we denote with  $\pi(C_n)$:
\begin{equation}
\label{equ:Pease}
\wht{n} =\prod_{k=1}^n\left((I_{2^{n-1}} \otimes \dft_2)\pi(C_n)\right).
\end{equation}

More generally, it is possible to enumerate all possible algorithms yielded by \eqref{equ:CT} by keeping the different values of $p$ and $q$ used in a partition tree \cite{Johnson:00}. The permutations obtained in this case are the identity, stride-permutations, and the permutations obtained by composition and Kronecker product of these. The group of $\emph{linear permutations}$ contains these \cite{Pueschel:09}.

\mypar{Linear permutation} If $Q$ is an $n\times n$ invertible bit-matrix ($Q\in \gl_n(\mathbb F_2)$), we denote with $\pi(Q)$ the associated linear permutation, i.e. the permutation on $2^n$ points that maps the element indexed by $0\leq i<2^n$ to the index $j$ satisfying $j_b=Qi_b$. As an example, the matrix
\[
C_n=\begin{pmatrix}&1&&\\&&\ddots&\\&&&1\\1&&&\end{pmatrix},
\] 
rotates the bits up, and its associated linear permutation is the perfect shuffle, hence the notation we used for it, $\pi(C_n)$.

\mypar{Sequence of linear permutations} In the rest of this article, we consider a sequence of $n+1$ invertible $n\times n$ bit-matrices $P=(P_0,P_1,\dots, P_n)$, and the computation
\[
W(P)=\pi(P_0)\cdot(I_{2^{n-1}} \otimes \dft_2)\cdot\pi(P_1)\cdot(I_{2^{n-1}} \otimes \dft_2)\dots\pi(P_{n-1})\cdot(I_{2^{n-1}} \otimes \dft_2)\cdot\pi(P_n).
\]

Note that we do not assume a priori that $P$ is such that $W(P) = \wht{n}$. In fact, we denote this subset, the set of \emph{$\dft_2$-based linear fast WHT algorithms} with $\mathcal P$: 
\[
\mathcal P =\left \{ P=(P_0,P_1,\dots, P_{n}) \text{ with } P_i \in \gl_n(\mathbb F_2) \mid W(P) = \wht{n}\right\}.
\]

\mypar{Product of matrices}
The product of the matrices $P_iP_{i+1}\cdots P_{j-1}P_j$ appears multiple times in the rest of this document, and we denote it therefore with $P_{i:j}$:
\[
P_{i:j} =\prod_{k=i}^j{P_k}.
\]
For convenience, we extend this notation by defining $P_{i:j}=I_n$ if $j<i$.

\mypar{Spreading matrix}
Similarly, the following matrix $X$ is recurring throughout this article:
\[
X=\begin{pmatrix}P_{0:n-1}1_b&P_{0:n-2}1_b& \cdots &P_{0:1}1_b&P_01_b \end{pmatrix}.
\]
Note that $1_b=\begin{pmatrix}0&\dots&0&1\end{pmatrix}^T$. Thus, $X$ results of the concatenation of the rightmost columns of the matrices $P_{0:n-1},\dots,P_{0}$. We will refer to this matrix as the \emph{spreading matrix}, as we will see that its invertibility is a necessary and sufficient condition for $W(P)$ to have no zero elements\footnote{It can be shown that a row of $W(P)$ contains at most $2^{\rank X}$ non-zero elements.}.

\subsection{Problem statement}
A na{\" i}ve approach to check if $P\in\mathcal P$ would compute $W(P)$ and compare it against $\wht{n}$. Therefore, it would perform $2n+1$ multiplications of $2^n\times 2^n$ matrices, and would have a complexity in $O(n\cdot 2^{3n})$ arithmetic operations. Our objective is to derive an equivalent set of conditions that can be checked with a polynomial complexity.

\subsection{\boldmath Characterization of WHT algorithms} Theorem~\ref{trm:characterisation} provides a necessary and sufficient set of conditions on a sequence of linear permutations such that the corresponding algorithm computes a WHT. A proof of this theorem is given in Section~\ref{sec:trm1}.
\begin{theor}\label{trm:characterisation}
$P \in \mathcal P$ if and only if the following conditions are satisfied:
\begin{itemize}
\item The product of the matrices satisfies
\begin{equation}\label{equ:P0n}
P_{0:n}=XX^T.
\end{equation}
\item The rows of the inverse of the spreading matrix are the last rows of the matrices $P_{0}^{-1}, \dots, P_{0:n-1}^{-1}$:
\begin{equation}\label{equ:Xinv}
X^{-1}=\begin{pmatrix}P_{0:n-1}^{-T}1_b&P_{0:n-2}^{-T}1_b    &\dots&P_{0:1}^{-T}1_b&P_0^{-T}1_b\end{pmatrix}^T.
\end{equation}

\end{itemize}
This set of conditions is minimal: there are counterexamples that do not satisfy one condition, while satisfying the other.
\end{theor}
\mypar{Cost} With this set of conditions, checking if a given sequence $P$ corresponds to a WHT requires $O(n^4)$ arithmetic operations.

\subsection{\boldmath Enumeration of WHT algorithms} Corollary~\ref{trm:enumerating} is the main contribution of this article. It allows to enumerate all linear fast WHT algorithms for a given $n$. For instance, all the matrices corresponding to the case $n=2$ are listed in Table~\ref{tab:comp}, and the corresponding dataflows in Fig.~\ref{fig:H2}. Section~\ref{sec:trm2} gives a proof of this corollary.
\begin{cor}
\label{trm:enumerating}
$P \in \mathcal P$ if and only if there exist $B\in \gl_n(\mathbb F_2)$ and $(Q_1, \cdots, Q_n)\in \left (\gl_{n-1}(\mathbb F_2)\right)^n$ such that
\begin{equation}
\label{equ:cor}
P_i=\begin{cases}
B\cdot\begin{pmatrix}Q_1&\\&1\end{pmatrix},&\text{ for }i=0,\\
\begin{pmatrix}Q_i^{-1}&\\&1\end{pmatrix}\cdot C_n\cdot\begin{pmatrix}Q_{i+1}&\\&1\end{pmatrix},&\text{ for }0<i<n,\\
\begin{pmatrix}Q_n^{-1}&\\&1\end{pmatrix}\cdot C_n\cdot B^T,&\text{ for }i=n.\\
\end{cases}
\end{equation}
\end{cor}

\begin{table}\small\center
\begin{tabular}{@{}llllll@{}}
\toprule
Ref.&$P_0$&$P_1$&$P_2$&$P_{0:n}$&$X$\\
\midrule
(a)&$\begin{pmatrix}1&\\&1\\\end{pmatrix}$ &$\begin{pmatrix}&1\\1&\\\end{pmatrix}$& $\begin{pmatrix}&1\\1&\\\end{pmatrix}$& $\begin{pmatrix}1&\\&1\\\end{pmatrix}$& $\begin{pmatrix}1&\\&1\\\end{pmatrix}$\\
(b)&$\begin{pmatrix}&1\\1&\\\end{pmatrix}$ &$\begin{pmatrix}&1\\1&\\\end{pmatrix}$& $\begin{pmatrix}1&\\&1\\\end{pmatrix}$& $\begin{pmatrix}1&\\&1\\\end{pmatrix}$& $\begin{pmatrix}&1\\1&\\\end{pmatrix}$\\
(c)&$\begin{pmatrix}1&\\1&1\\\end{pmatrix}$ &$\begin{pmatrix}&1\\1&\\\end{pmatrix}$& $\begin{pmatrix}&1\\1&1\\\end{pmatrix}$& $\begin{pmatrix}1&1\\1&\\\end{pmatrix}$& $\begin{pmatrix}1&\\1&1\\\end{pmatrix}$\\
(d)&$\begin{pmatrix}1&1\\&1\\\end{pmatrix}$ &$\begin{pmatrix}&1\\1&\\\end{pmatrix}$& $\begin{pmatrix}1&1\\1&\\\end{pmatrix}$& $\begin{pmatrix}&1\\1&1\\\end{pmatrix}$& $\begin{pmatrix}1&1\\&1\\\end{pmatrix}$\\
(e)&$\begin{pmatrix}1&1\\1&\\\end{pmatrix}$ &$\begin{pmatrix}&1\\1&\\\end{pmatrix}$& $\begin{pmatrix}1&\\1&1\\\end{pmatrix}$& $\begin{pmatrix}&1\\1&1\\\end{pmatrix}$& $\begin{pmatrix}1&1\\1&\\\end{pmatrix}$\\
(f)&$\begin{pmatrix}&1\\1&1\\\end{pmatrix}$ &$\begin{pmatrix}&1\\1&\\\end{pmatrix}$& $\begin{pmatrix}1&1\\&1\\\end{pmatrix}$& $\begin{pmatrix}1&1\\1&\\\end{pmatrix}$& $\begin{pmatrix}&1\\1&1\\\end{pmatrix}$\\
\bottomrule
\end{tabular}
\caption{Matrices of all $\dft_2$-based linear fast algorithms computing $\wht{2}$, and the corresponding product of matrices ($P_{0:n}$) and spreading matrix ($X$). The first line corresponds to the Pease algorithm, the second one to its transpose. These two algorithms are the only ones that can be obtained using \eqref{equ:CT}.}\label{tab:comp}
\end{table}
\begin{figure}\centering
\subfigure[]{\includegraphics[scale=1]{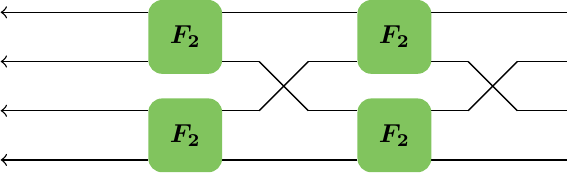}}
\subfigure[]{\includegraphics[scale=1]{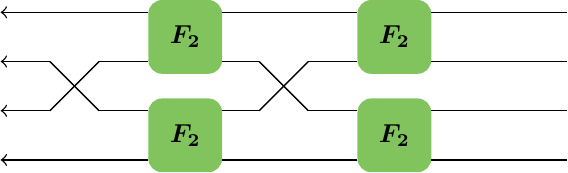}}
\subfigure[]{\includegraphics[scale=1]{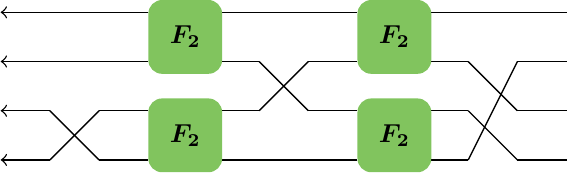}}
\subfigure[]{\includegraphics[scale=1]{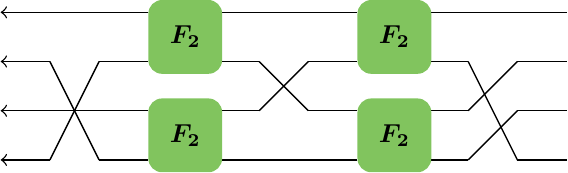}}
\subfigure[]{\includegraphics[scale=1]{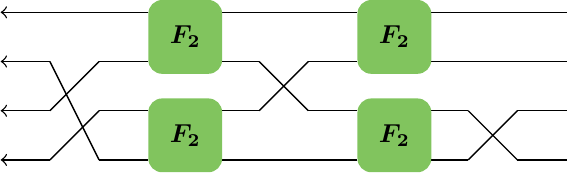}}
\subfigure[]{\includegraphics[scale=1]{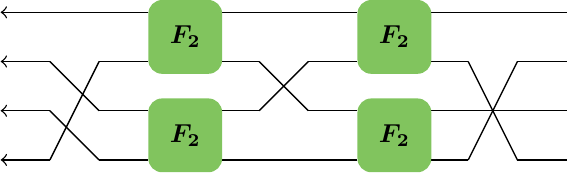}}
\caption{Dataflow of all $\dft_2$-based linear fast algorithms computing $\wht{2}$. The letters correspond to the references in Table~\ref{tab:comp}}
\label{fig:H2}
\end{figure}

\mypar{Size of $\mathcal P$} For a given $n$, Corollary~\ref{trm:enumerating} shows a direct map between $\mathcal P$ and $\gl_{n}(\mathbb F_2)\times \left (\gl_{n-1}(\mathbb F_2)\right)^n$. This yields the number of linear fast algorithms that compute a WHT:
\begin{align*}
|\mathcal P|&=\prod_{i=0}^{n-1}(2^n-2^i)\prod_{i=0}^{n-2}(2^{n-1}-2^i)^n\\
&=(2^{n+1}-2)\prod_{i=0}^{n-2}(2^{n-1}-2^i)^{n+1}.\\
\end{align*}
Table~\ref{tab:size} lists the first few values of $|\mathcal P|$. In practice, even for relatively small $n$, this size makes any exhaustive search based approach on $\mathcal P$ illusive. Storing the bit-matrices of all $\dft_2$-based linear fast algorithms computing $\wht{4}$ requires $40$GB, and there are more algorithms computing $\wht{6}$ than atoms in the earth.

\begin{table}\center
\begin{tabular}{@{}c|cccccccc@{}}
$n$&$1$&$2$&$3$&$4$&$5$&$6$&$7$&$8$\\
\midrule
$|\text{derived from \eqref{equ:CT}}|$ &$1$ &$2$ &$6$ &$24$ &$112$ &$568$ &$3032$ &$16768$\\
$|\mathcal P\cap \mathcal S_n^{n+1}|$ &$1$ &$2$ &$48$ &$31104$ &$\approx10^{9}$ &$\approx2\cdot10^{15}$ &$\approx5\cdot10^{23}$ &$\approx2\cdot10^{34}$\\
$|\mathcal P|$ &$1$ &$6$ &$18 144$ &$\approx4\cdot10^{12}$ &$\approx4\cdot10^{27}$ &$\approx10^{51}$ &$\approx7\cdot10^{84}$ &$\approx4\cdot10^{130}$\\
\end{tabular}
\caption{Number of $\dft_2$-based linear fast algorithms $|\mathcal P|$, number of $\dft_2$-based bit-index-permuted fast algorithms $|\mathcal P\cap \mathcal S_n^{n+1}|$ and number of algorithms that can be derived from \eqref{equ:CT} \cite{Johnson:00} for a given $n$.}\label{tab:size}
\end{table}

\subsection{Other transforms and other permutations}
In this paper, we consider linearly permuted fast algorithms that compute the unscaled and \emph{naturally ordered} WHT. In this section, we discuss about related algorithms, and one other set of permutations.

\mypar{Walsh transform} The \emph{sequency ordered} version, represented by Walsh matrix, only differs by a \emph{bit-reversal} permutation of its outputs. The bit-reversal is the permutation $\pi(J_n)$ that flips the bits of the indices:
\[
J_n=\begin{pmatrix}&&1\\&\iddots\\1\end{pmatrix}.
\]
All the results obtained in this paper can be used for the Walsh matrix, after multiplying $P_0$ by $J_n$ on the left. Particularly, the number of $\dft_2$-based linear fast Walsh transform is the same as for the WHT.

\mypar{Orthogonal WHT} The WHT can be made \emph{orthogonal} by scaling it by a factor of $2^{-n/2}$. Algorithms performing this transform can be obtained with our technique by using the orthogonal $\dft_2$, i.e. the butterfly scaled by a factor of $1/\sqrt{2}$.

\mypar{Bit-index permutations} Bit-index permutations are linear permutations for which the bit-matrix is itself a permutation matrix. As these are a subset of linear permutations, all the results presented here apply, and particularly, the bit-index-permuted algorithms can be enumerated using Corollary~\ref{trm:enumerating}, but using $B\in \mathcal S_n(\mathbb F_2)$ and $(Q_1,\cdots,Q_n)\in \left (\mathcal S_{n-1}(\mathbb F_2)\right)^n$. The number of these algorithms for a given $n$ is
\[
|\mathcal P\cap S_n^{n+1}|=n((n-1)!)^{n+1},
\]
Table~\ref{tab:size} gives its first few values.

\section{Proof of Theorem~\ref{trm:characterisation}}\label{sec:trm1}
In this section, we provide a proof of Theorem ~\ref{trm:characterisation}. The main idea of this proof consists in deriving a general expression of $W(P)$, assuming only that $W(P)$ has its first row and its first column filled with $1$s (Lemma~\ref{lem:dp}). Then, we match this expression with the definition of a WHT to derive necessary and necessary conditions for an algorithm to compute a WHT. Before that, in Lemma~\ref{lem:d}, we derive some consequences of the invertibility of the spreading matrix $X$, particularly on the non-zero elements of $W(P)$. In Lemma~\ref{lem:li}, we provide a necessary and sufficient condition for $W(P)$ to have a $1$-filled first row and column. We begin by defining concepts that will be used throughout this section.

\mypar{Stage of an algorithm} We will refer to the stage \emph{$k$} as an array of $\dft_2$ composed with the linear permutation associated with $P_k$: $(I_{2^{n-1}} \otimes \dft_2)\pi(P_k)$. Due to the ordering of evaluation, from right to left, we call the stage $n$ is the rightmost stage of an algorithm, and the stage $k$ is on the left of the stage $k+1$. We denote with $W_k(P)$ the matrix corresponding to the output on the left of stage $k$, i.e.,
\[
W_k(P)=(I_{2^{n-1}}\otimes \dft_2)\pi(P_k)\dots(I_{2^{n-1}}\otimes \dft_2)\pi(P_n).
\]
As a consequence, $W(P)=\pi(P_0)W_1(P)$. For practical reasons, we extend this definition for $k=n+1$ by considering that $W_{n+1}(P)=I_n$.

\mypar{\boldmath Outputs depending on input $i$ on the left of stage $k$}
For $0\leq i<2^n$, we denote with $\mathcal D_k(i)$ the set of the outputs of the $k^{th}$ stage of the algorithm that depend on the $i^{th}$ input:
\[
\mathcal D_k(i) =\left \{ j_b \mid W_k(P)[j,i]\neq 0 \right\}\subseteq \mathbb F_2^n.
\]

The dependency of the whole algorithm on the input $i$ is denoted with $\mathcal D(i)$:
\[
\mathcal D(i) =\left \{ j_b \mid W(P)[j,i]\neq 0 \right\}.
\]
Similarly, we denote with $\mathcal D_k^+(i)$ (resp. $\mathcal D_k^-(i)$) the set of the indices of the outputs for which 
\[
\begin{array}{l}
\mathcal D_k^+(i) =\left \{ j_b \mid W_k(P)[j,i]=1 \right\},\text{ and }\mathcal D_k^-(i) =\left \{ j_b \mid W_k(P)[j,i]=-1 \right\}.
\end{array}
\]
\subsection{Invertibility of the spreading matrix}
In the following lemma, we justify the name ``spreading matrix'' that we use for $X$, by showing that its invertibility conditions the ``spread'' of non-zero elements through the rows of $W(P)$, and provide some other consequences we will use later.
\begin{lem}\label{lem:d}
All the outputs of the algorithm depend on the first input, i.e. $\mathcal D(0)=\mathbb F_2^n$ if and only if the spreading matrix $X$ is invertible. In this case, for every $0\leq i<2^n$ and $1\leq k\leq n$,
\begin{itemize}
\item The set of dependency at the $k^\text{th}$ stage on the input $i$ is
\begin{equation}
\label{equ:dep}
\mathcal D_{k}(i) = (P_{k}\mathcal D_{k+1}(i))\cup (P_{k}\mathcal D_{k+1}(i)+1_b)
\end{equation}
\item The non-zero elements of $W_k(P)$ are either $1$ or $-1$:
\begin{equation}
\label{equ:unionpm}
\mathcal D_{k}(i)=\mathcal D_{k}^+(i)\cup \mathcal D_{k}^-(i).
\end{equation}
\item The $k^\text{th}$ stage modifies the set of dependencies such that: 
\begin{equation}
\label{equ:recDkp}
\mathcal D_{k}^+(i)=(P_k\mathcal D_{k+1}^+(i)+1_b)\cup\{ j_b\in P_k\mathcal D_{k+1}^+(i)\mid j_b^T 1_b=0 \}\cup\{ j_b\in P_k\mathcal D_{k+1}^-(i)\mid j_b^T 1_b=1 \}.
\end{equation}
\end{itemize}
\end{lem}
\begin{proof}
First, we assume that $\mathcal D(0)=\mathbb F_2^n$, and show that $X$ is invertible. For a given $j$, the two outputs $j_b$ and $j_b+1_b$ of a $\dft_2$ may have a dependency on the first input only if at least one of the two signals $j_b$ and $j_b+1_b$ that arrive on this $\dft_2$ depends on that input.
Therefore, we have
\[
\mathcal D_{k}(0) \subseteq P_{k}\mathcal D_{k+1}(0)\cup (P_{k}\mathcal D_{k+1}(0)+1_b).
\]
We now prove by induction that
\[
\mathcal D_{k}(0) \subseteq \lgen 1_b, P_k1_b, P_{k:k+1}1_b,\dots,P_{k:n-1}1_b\rgen.
\]
We already have that $\mathcal D_{n+1}(0) =\{0_b\}$. Assuming that the result holds at rank $k+1$, we have:
\[
\begin{array}{rcl}
\mathcal D_{k}(0) &\subseteq& P_{k}\mathcal D_{k+1}(0)\cup (P_{k}\mathcal D_{k+1}(0)+1_b)\\
 &\subseteq& P_k\lgen 1_b, P_{k+1}1_b,\dots,P_{k+1:n-1}1_b\rgen\cup (P_k\lgen 1_b, P_{k+1}1_b,\dots,P_{k+1:n-1}1_b\rgen+1_b)\\
  &=& \lgen P_k1_b, P_{k:k+1}1_b,\dots,P_{k:n-1}1_b\rgen\cup (\lgen P_{k}1_b,P_{k:k+1}1_b,\dots,P_{k:n-1}1_b\rgen+1_b)\\
  &=& \lgen 1_b, P_k1_b, P_{k:k+1}1_b,\dots,P_{k:n-1}1_b\rgen.
\end{array}
\]
Which yields the result. As a consequence, we have
\[
\mathcal D(0)=P_0\mathcal D_1(0) \subseteq P_0\lgen 1_b,P_11_b,\dots,P_{1:n-1}1_b\rgen=\im X.
\]
As $\mathcal D(0)=\mathbb F_2^n$, we have $\mathbb F_2^n\subseteq \im X$, and $X$ is therefore invertible.

Conversely, we now assume that $X$ is invertible, and prove by induction that the set of outputs after stage $k$ that depend on the $i^\text{th}$ input is
\begin{equation}
\label{equ:Dk}
\mathcal D_{k}(i) = P_{k:n}i_b+\lgen 1_b, P_k1_b, P_{k:k+1}1_b,\dots,P_{k:n-1}1_b\rgen.
\end{equation}
We already have $\mathcal D_{n+1}(i)=\{i_b\}$. Assuming \eqref{equ:Dk} for $k+1$, and considering an element
\[
j_b\in P_k\mathcal D_{k+1}(i) = P_{k:n}i_b+\lgen P_{k}1_b,P_{k:k+1}1_b, \dots,P_{k:n-1}1_b\rgen,
\]
we have $j_b + P_{k:n}i_b\in \lgen P_{k}1_b, \dots, P_{k:n-1}1_b\rgen$. 
As $X$ is invertible, so is the matrix 
\[
P_{0:k-1}^{-1}X=\begin{pmatrix}P_{k:n-1}1_b&\cdots&P_{k}1_b&1_b&P_{k-1}^{-1}1_b&\cdots&P_{1:k-1}^{-1}1_b\end{pmatrix}.
\]
Its columns are linearly independent, and particularly,  $1_b\notin \lgen P_{k}1_b, \dots,P_{k:n-1}1_b\rgen$. Therefore, $j_b + P_{k:n}i_b+1_b\notin \lgen P_{k}1_b, \dots,P_{k:n-1}1_b\rgen$, thus $j_b+1_b\notin P_k \mathcal D_{k+1}(i)$. This means that if a signal $j_b$ that arrives on a $\dft_2$ depends on an input $i$ ($j_b\in P_k\mathcal D_{k+1}(i)$), the other signal ($j_b+1_b$) doesn't. As the output of a $\dft_2$ is the sum (resp. the difference) of these, and that an input $i$ never appears on both terms of this operation, the dependency of both outputs on inputs is the union of the dependencies of the signals that arrive to this $\dft_2$. This yields \eqref{equ:dep}, and a direct computation shows that
\[
\begin{array}{rcl}
\mathcal D_{k}(i) &=& P_{k}\mathcal D_{k+1}(i)\cup (P_{k}\mathcal D_{k+1}(i)+1_b)\\
&=& (P_{k:n}i_b+\lgen P_k1_b,\dots,P_{k:n-1}1_b\rgen)\cup (P_{k:n}i_b+\lgen P_k1_b, \dots,P_{k:n-1}1_b\rgen+1_b)\\
&=& P_{k:n}i_b+\lgen 1_b, P_k1_b,\dots,P_{k:n-1}1_b\rgen.
\end{array}
\]
This yields \eqref{equ:Dk}, and $\mathcal D(0)=\mathbb F_2^n$ as a direct consequence. To be more precise, if \eqref{equ:unionpm} is satisfied at rank $k+1$, a signal $j_b$ depending on the input $i$ that arrives on the $\dft_2$ array of the $k^\text{th}$ stage is in one of these cases:
\begin{itemize}
\item $j_b\in P_k\mathcal D_{k+1}^+(i)$, and $j$ arrives on top of the $\dft_2$ ($j$ is even, i.e. $j_b^T 1_b=0$). In this case, both outputs of this $\dft_2$ depend ``positively'' on $i$: $\{j_b, j_b+1_b\} \subseteq \mathcal D_{k}^+(i)$.
\item $j_b\in P_k\mathcal D_{k+1}^-(i)$, and $j$ arrives on top of the $\dft_2$ ($j$ is even, i.e. $j_b^T 1_b=0$). In this case, both outputs of this $\dft_2$ depend ``negatively'' on $i$: $\{j_b, j_b+1_b\} \subseteq \mathcal D_{k}^-(i)$.
\item $j_b\in P_k\mathcal D_{k+1}^+(i)$, and $j$ arrives on the bottom of the $\dft_2$ ($j$ is odd, i.e. $j_b^T 1_b=1$). In this case, the top output depends ``positively'' on $i$: $j_b+1_b \in \mathcal D_{k}^+(i)$, and the bottom output ``negatively'': $j_b \in \mathcal D_{k}^-(i)$.
\item $j_b\in P_k\mathcal D_{k+1}^-(i)$, and $j$ arrives on the bottom of the $\dft_2$ ($j$ is odd, i.e. $j_b^T 1_b=1$). In this case, the top output depends ``negatively'' on $i$: $j_b+1_b \in \mathcal D_{k}^-(i)$, and the bottom output ``positively'': $j_b \in \mathcal D_{k}^+(i)$.
\end{itemize}
This yields \eqref{equ:recDkp} and \eqref{equ:unionpm} at rank $k$. As we have as well $\mathcal D_{n+1}^+(i)\cup \mathcal D_{n+1}^-(i)=\{i_b\}\cup \emptyset=\mathcal D_{n+1}(i)$, \eqref{equ:unionpm} holds for all $k$.
\end{proof}

\subsection{About condition \eqref{equ:Xinv}}
As mentioned earlier, we will derive a general expression for $W(P)$, assuming that it has its first row and columns filled with $1$s. In the following theorem, we provide an equivalent condition for this assumption.
\begin{lem}
\label{lem:li}
The following propositions are equivalent:
\begin{itemize}
\item The first row and the first column of $W(P)$ contain only $1$s:
\begin{equation}
\label{equ:d+0}
\mathcal D^+(0)=\mathbb F_2^n, \text{ and}
\end{equation}
\begin{equation}
\label{equ:d+i}
0_b\in \mathcal D^+(i),\text{ for } 0\leq i<2^n.
\end{equation}
\item The spreading matrix satisfies \eqref{equ:Xinv}.
\item No sequential product $P_{k:\ell}$ of the central matrices has a $1$ in the bottom right corner, and the same holds for the inverse of $P_{k:\ell}$:
\begin{equation}\label{equ:central}
1_b^T P_{k:\ell}1_b=0, \text{ for } 0< k \leq \ell < n, \text{ and}
\end{equation}
\begin{equation}\label{equ:centralinv}
1_b^T P_{k:\ell}^{-1}1_b=0, \text{ for } 0< k \leq \ell < n.
\end{equation}
\end{itemize}
\end{lem}
\begin{proof}
We start by showing the equivalence between the second and the third proposition. We consider the canonical basis $(e_1,\dots,e_n)$ of $\mathbb F_2^n$. We have, for all $0< k,k'\leq n$,
\begin{align*}
e_k^T\begin{pmatrix}P_{0:n-1}^{-T}1_b&P_{0:n-2}^{-T}1_b    &\dots&P_{0:1}^{-T}1_b&P_0^{-T}1_b\end{pmatrix}^TXe_{k'}&=(P_{0:n-k}^{-T}1_b)^TP_{0:n-k'}1_b\\
&=1_b^TP_{1:n-k}^{-1}P_{1:n-k'}1_b\\
&=\begin{cases}
1 & \text{ if } k=k',\text{ or }\\
1_b^T P_{n-k-1:n-k'}1_b & \text{ if } k>k',\text{ or }\\
1_b^T P^{-1}_{n-k'-1:n-k}1_b &  \text{ if } k'>k.\\
\end{cases}
\end{align*}
The nullity of the two last cases is equivalent to \eqref{equ:Xinv} on one side, and \eqref{equ:central} and \eqref{equ:centralinv} on the other side.

We now consider the first proposition. We assume first that \eqref{equ:d+0} holds, and will show that it implies \eqref{equ:central}. As $\mathbb F_2^n=\mathcal D^+(0)\subseteq \mathcal D(0)$, we can use the results of Lemma~\ref{lem:d}. Using \eqref{equ:dep}, \eqref{equ:unionpm} and  \eqref{equ:recDkp}, we have
\begin{align*}
\mathcal D_{k}^-(i)=&\mathcal D_{k}(i)\setminus\mathcal D_{k}^+(i)\\
=&\left(P_{k}\mathcal D_{k+1}(i)\cup (P_{k}\mathcal D_{k+1}(i)+1_b)\right)\setminus\\
&\left((P_k\mathcal D_{k+1}^+(i)+1_b)\cup\{ j_b\in P_k\mathcal D_{k+1}^+(i)\mid j_b^T 1_b=0 \}\cup\{ j_b\in P_k\mathcal D_{k+1}^-(i)\mid j_b^T 1_b=1 \}\right)\\
\supseteq&\left(P_{k}\mathcal D_{k+1}(i)+1_b)\right)\setminus\left(P_k\mathcal D_{k+1}^+(i)+1_b\right)\\
=&P_k\mathcal D_{k+1}^-(i)+1_b\\
\end{align*}
Therefore, the number of outputs depending ``negatively'' on a given input $i$ can only increase within a stage:
\[
|\mathcal D_{k}^-(i)|\geq |\mathcal D_{k+1}^-(i)|.
\]

Particularly, for the first input, this means that $|\mathcal D^-(0)|=|\mathcal D_{1}^-(0)|\geq\dots\geq|\mathcal D_{n}^-(0)|$. As $\mathcal D^-(0)=\emptyset$, we have $\mathcal D_{k}^+(0)=\mathcal D_{k}(0)$ for all $k$. 
Using again \eqref{equ:dep}, \eqref{equ:unionpm} and \eqref{equ:recDkp}, we have, for all $k$,
\begin{align*}
\emptyset=&\mathcal D_{k}^-(0)\\
=&\left(P_{k}\mathcal D_{k+1}(0)\cup (P_{k}\mathcal D_{k+1}(0)+1_b)\right)\setminus\\
&\left((P_k\mathcal D_{k+1}^+(0)+1_b)\cup\{ j_b\in P_k\mathcal D_{k+1}^+(0)\mid j_b^T 1_b=0 \}\right)\\
=& \{ j_b\in P_k\mathcal D_{k+1}^+(0)\mid j_b^T 1_b=1 \},
\end{align*}
Therefore, $j_b^T1_b=0$ for all $k$ and all $j_b\in P_k\mathcal D_{k+1}^+(0)=P_k\mathcal D_{k+1}(0)=\lgen P_k1_b, P_kP_{k+1}1_b,\dots,P_{k:n-1}1_b\rgen$, which yields \eqref{equ:central}.

We now consider that \eqref{equ:d+i} is satisfied. This means that \eqref{equ:d+0} holds for $W(P)^T=W((P_n^{-1},\dots,P_0^{-1}))$, and the same computation yields \eqref{equ:centralinv}.

Finally, we suppose that \eqref{equ:central} and \eqref{equ:centralinv} hold (and therefore \eqref{equ:Xinv}, which allows to use Lemma~\ref{lem:d}). For $i=0$, \eqref{equ:recDkp} becomes
\[
\mathcal D_{k}^+(0)=(P_k\mathcal D_{k+1}^+(0)+1_b)\cup\{ j_b\in P_k\mathcal D_{k+1}^+(0)\mid j_b^T 1_b=0 \}\cup\{ j_b\in P_k\mathcal D_{k+1}^-(0)\mid j_b^T 1_b=1 \}.
\]
However, \eqref{equ:central} yields that $j_b^T1_b=0$ in all cases. Therefore, we have
\[
\mathcal D_{k}^+(0)=(P_k\mathcal D_{k+1}^+(0)+1_b)\cup\{ j_b\in P_k\mathcal D_{k+1}^+(0)\},
\]
and a direct induction yields \eqref{equ:d+0}. The same argument with \eqref{equ:centralinv} and $W(P)^T$ yields \eqref{equ:d+i}.
\end{proof}

\subsection{\boldmath General expression of $W(P)$}
When $X$ is invertible, we have $\mathcal D(i)=\mathcal D^+(i)\cup \mathcal D^-(i)$, which means that $W(P)$ is entirely determined by $\mathcal D^+$, for which we derive an expression in the following lemma.
\begin{lem}\label{lem:dp}
If the spreading matrix satisfies \eqref{equ:Xinv}, then
\begin{equation}\label{equ:Dp}
\mathcal D^+(i)=\ker \left(i_b^T P_{0:n}^T(XX^T)^{-1}\right)=\left\{j_b\mid i_b^T P_{0:n}^T(XX^T)^{-1}j_b=0\right\}.
\end{equation}
\end{lem}
\begin{proof}
We assume that $X$ satisfies \eqref{equ:Xinv}. As $X$ is invertible, the results of Lemma~\ref{lem:d} and \ref{lem:li} can be used. Additionally, for $1\leq k\leq n+1$, the vectors $\{1_b, P_k1_b, P_{k:k+1}1_b,\dots,P_{k:n-1}1_b\}$ are linearly independent (as columns of the invertible matrix $P_{0:k-1}^{-1}X$), and we define on the ($n-k+1$)-dimensional space spawn by these vectors the linear mapping $f_k$ as follows:
\[
\begin{array}{crcl}
  & 1_b & \mapsto & P^T_{k:n}1_b \\
  & P_k1_b & \mapsto & P^T_{k+1:n}1_b \\
   & \vdots & \vdots & \vdots \\
    & P_{k:n-1}1_b & \mapsto & P^T_n1_b. \\
\end{array}
\]
This mapping satisfies, for $1\leq k\leq n$, and for all $x$ in the domain of $f_{k+1}$,
\begin{equation}
\label{equ:fkfkp1}
f_{k}(P_kx)=f_{k+1}(x).
\end{equation}
Additionally, for $k=1$, this mapping is defined over $\mathbb F_2^n$, and satisfies, for all $0\leq x<2^n$,
\begin{equation}
\begin{array}{rcl}
f_1(x_b)&=&\begin{pmatrix}P_{n}^{T}1_b&P_{n-1:n}^{T}1_b &\cdots&P_{1:n}^T1_b\end{pmatrix}\begin{pmatrix}P_{1:n-1}1_b&P_{1:n-2}1_b   \cdots&1_b\end{pmatrix}^{-1}x_b\\
&=&P_{0:n}^T\begin{pmatrix}P_{0:n-1}^{-T}1_b&P_{0:n-2}^{-T}1_b&\cdots&P_{0}^{-T}1_b\end{pmatrix}\begin{pmatrix}P_{0:n-1}1_b&P_{0:n-2}1_b&\cdots&P_{0}1_b\end{pmatrix}^{-1}P_0x_b\\
&=&P_{0:n}^T(XX^{T})^{-1}P_0x_b.\\
\end{array}
\end{equation}

We define as well the vector $s_k$:
\[
s_k=1_b+P_k1_b+P_{k:k+1}1_b+\dots+P_{k:n-1}1_b.
\]

We will now prove by induction that, for $1\leq k\leq n+1$, 
\begin{equation}
\label{equ:Dkp}
\mathcal D_{k}^+(i)=\{j_b\in \mathcal D_{k}(i)\mid i_b^Tf_{k}(j_b+P_{k:n}i_b+s_k)=0\}.
\end{equation}
We have, for $k=n+1$,
\[
\begin{array}{rcl}
\mathcal D_{n+1}^+(i) & =& \{i_b\}\\
&=&\{j_b\in \mathcal D_{n+1}(i)\mid i_b^Tf_{n+1}(j_b+I_ni_b+0_b)=0\}.
\end{array}
\]
If we assume that \eqref{equ:Dkp} is satisfied at rank $k+1$, i.e $\mathcal D_{k+1}^+(i)=\{j_b\in \mathcal D_{k+1}(i)\mid i_b^T f_{k+1}(j_b+P_{k+1:n}i_b+s_{k+1})=0\},$ we have, using \eqref{equ:fkfkp1},
\[
\begin{array}{rcl}
P_k\mathcal D_{k+1}^+(i)&=&\{j_b\in P_k\mathcal D_{k+1}(i)\mid i_b^Tf_{k+1}(P_k^{-1}(j_b+P_{k:n}i_b+P_ks_{k+1}))=0\}\\
&=&\{j_b\in P_k\mathcal D_{k+1}(i)\mid i_b^Tf_{k}(j_b+P_{k:n}i_b+P_ks_{k+1})=0\}.
\end{array}
\]
Using \eqref{equ:unionpm}, we directly get
\[
P_k\mathcal D_{k+1}^-(i)=\{j_b\in P_k\mathcal D_{k+1}(i)\mid i_b^Tf_{k}(j_b+P_{k:n}i_b+P_ks_{k+1})=1\}.
\]
We can now compute $\mathcal D_{k}^+$ using \eqref{equ:recDkp}:
\[
\begin{array}{rcl}
\mathcal D_{k}^+(i)&=&(P_k\mathcal D_{k+1}^+(i)+1_b)\cup\{ j_b\in P_k\mathcal D_{k+1}^+(i)\mid j_b^T 1_b=0 \}\cup\{ j_b\in P_k\mathcal D_{k+1}^-(i)\mid j_b^T1_b=1 \}\\
&=&(P_k\mathcal D_{k+1}^+(i)+1_b)\cup\{j_b\in P_k\mathcal D_{k+1}(i)\mid i_b^Tf_{k}(j_b+P_{k:n}i_b+P_ks_{k+1})+j_b^T 1_b=0\}.\\
\end{array}
\]
The term $j_b^T1_b$ that appears for $j_b\in P_k\mathcal D_{k+1}(i)=P_{k:n}i_b+\lgen  P_k1_b, P_{k:k+1}1_b,\dots,P_{k:n-1}1_b\rgen$ satisfies, using \eqref{equ:central}, $j_b^T1_b=i_b^TP_{k:n}^T1_b=i_b^Tf_k(1_b)$. Therefore:
\[
\begin{array}{rcl}
\mathcal D_{k}^+(i)&=&(P_k\mathcal D_{k+1}^+(i)+1_b)\cup\{j_b\in P_k\mathcal D_{k+1}(i)\mid i_b^T f_{k}(j_b+P_{k:n}i_b+P_ks_{k+1})+i_b^T f_k(1_b)=0\}\\
&=&(P_k\mathcal D_{k+1}^+(i)+1_b)\cup\{j_b\in P_k\mathcal D_{k+1}(i)\mid i_b^Tf_{k}(j_b+P_{k:n}i_b+s_k)=0\}\\
&=&\{j_b\in P_k\mathcal D_{k+1}(i)+1_b\mid i_b^Tf_{k}(j_b+P_{k:n}i_b+s_k)=0\} \cup\\
&&\{j_b\in P_k\mathcal D_{k+1}(i)\mid i_b^Tf_{k}(j_b+P_{k:n}i_b+s_k)=0\}\\
&=&\{j_b\in \mathcal D_{k}(i)\mid i_b^Tf_{k}(j_b+P_{k:n}i_b+s_k)=0\},
\end{array}
\]
which yields the result.

We can now provide a first expression for $\mathcal D^+$, using \eqref{equ:fkfkp1}:
\[
\begin{array}{rcl}
\mathcal D^+(i)&=&P_0\mathcal D^+_1(i)\\
&=&P_0\{j_b\mid i_b^Tf_1(j_b+P_{1:n}i_b+s_1)=0\}\\
&=&P_0\{j_b\mid i_b^TP_{0:n}^T(XX^T)^{-1}(P_0j_b+P_{0:n}i_b+P_0s_1)=0\}\\
&=&\{j_b\mid i_b^TP_{0:n}^T(XX^T)^{-1}(j_b+P_{0:n}i_b+P_0s_1)=0\}.\\
\end{array}
\]

The last step consists in showing that the mapping $g: i_b\mapsto i_b^TP_{0:n}^T(XX^T)^{-1}(P_{0:n}i_b+P_0s_1)$ is null. We consider the canonical basis $(e_1,\dots,e_{n})$ of $\mathbb F_2^n$, and a vector $x=\sum_i\lambda_iP_{0:n}^{-1}Xe_i$ of this space. A direct computation yields:
\begin{align*}
g(x)&=x^TP_{0:n}^T(XX^T)^{-1}(P_{0:n}x+P_0s_1)\\
&=(X^{-1}P_{0:n}x)^TX^{-1}(P_{0:n}x+P_0s_1)\\
&=\left(X^{-1}P_{0:n}\sum_i\lambda_iP_{0:n}^{-1}Xe_i\right)^TX^{-1}\left(P_{0:n}\sum_j\lambda_jP_{0:n}^{-1}Xe_j+\sum_k Xe_k\right)\\
&=\sum_i\lambda_ie_i^T\sum_j(\lambda_j+1)e_j\\
&=\sum_{i,j}\lambda_i(\lambda_j+1)e_i^Te_j\\
&=\sum_{i=j}\lambda_i(\lambda_j+1)e_i^Te_j+\sum_{i\neq j}\lambda_i(\lambda_j+1)e_i^Te_j\\
&=0.\\
\end{align*}
%
%
%
%
%
%
\end{proof}

\subsection{Proof of theorem \ref{trm:characterisation}}
The proof of theorem \ref{trm:characterisation} is now straightforward. 

If $P\in \mathcal P$, then $W(P)$ has its first column and row filled up with $1$s. Lemma~\ref{lem:li} ensures \eqref{equ:Xinv}, and we can use Lemma~\ref{lem:dp}. Identifying \eqref{equ:Dp} with the definition \eqref{equ:def} of $\wht{n}$ yields that $P_{0:n}^T=XX^T$, i.e. that \eqref{equ:P0n} is satisfied.

Conversely, if \eqref{equ:P0n} and  \eqref{equ:Xinv} are satisfied, $W(P)=\wht{n}$ is a direct consequence of Lemma~\ref{lem:dp}.

\section{Proof of Corollary~\ref{trm:enumerating}}\label{sec:trm2}
We first assume that $(P_0,\dots,P_n) \in \mathcal P$, and construct a set of matrices satisfying \eqref{equ:cor}. We define $B=X^{-1}$, and, by induction, the matrices 
\[
\tilde Q_i=\begin{cases}
B^{-1}\cdot P_0,&\text{ for }i=1,\\
C_n^{-1}\cdot \tilde Q_{i-1}\cdot P_{i-1},&\text{ for }1<i\leq n.\\
\end{cases}
\]
By construction, these matrices satisfy, for $0<i\leq n$, $\tilde Q_i=C_n^{1-i}X^{-1}P_{0:i-1}$. Using \eqref{equ:P0n}, we get:
\[
P_i=\begin{cases}
B\cdot \tilde Q_1,&\text{ for }i=0,\\
\tilde Q_i^{-1} \cdot C_n\cdot \tilde Q_{i+1},&\text{ for }0<i<n,\\
\tilde Q_n^{-1} \cdot C_n\cdot \tilde B^T,&\text{ for }i=n.\\
\end{cases}
\]
The last step consists in showing that, for $0<i\leq n$, there exists a matrix $Q_i\in\gl_{n-1}(\mathbb F_2)$ such that $\tilde Q_i=\begin{pmatrix}Q_{i}&\\&1\end{pmatrix}$, or equivalently, that $\tilde Q_i1_b=\tilde Q_i^T1_b=1_b$:
\begin{align*}
\tilde Q_i1_b&=C_n^{1-i} X^{-1}P_{0:i-1}1_b\\
&=C_n^{1-i} \begin{pmatrix}P_{0:n-1}1_b&\dots&P_{0:1}1_b&P_01_b\end{pmatrix}^{-1}P_{0:i-1}1_b\\
&=1_b.
\end{align*}
A similar computation, using \eqref{equ:Xinv}, shows that $\tilde Q_i^T1_b=1_b$.

Conversely, if a set of matrices satisfy \eqref{equ:cor}, a direct computation (with \cite{Johnson:90}) shows that \eqref{equ:P0n} and \eqref{equ:Xinv} are satisfied. Thus. $P\in \mathcal P$.

\section{Conclusion}
We described a method to characterize and enumerate fast linear WHT algorithms. As methods are known to implement optimally linear permutations on hardware \cite{Serre:16a}, a natural future work would consist in finding the best WHT algorithm that can be implemented for a given $n$. Additionally, it would be interesting to see if similar algorithms exist for other transforms, like the $\dft$.

\bibliographystyle{plain}
\bibliography{bib}
\end{document}